\documentclass{llncs}
\usepackage[T1]{fontenc}
\usepackage[english]{babel}

\usepackage{amsmath}
\usepackage{amssymb}
\usepackage{amsfonts,xspace,textcomp}

\usepackage{tikz}
\usetikzlibrary{automata}

\newcommand{\mb}{\text{\textblank}}
\newcommand{\ignore}[1]{}
\newcommand{\N}{\ensuremath{\mathbb{N}}\xspace}
\newcommand{\Z}{\ensuremath{\mathbb{Z}}\xspace}
\newcommand{\az}{\ensuremath{A^{\Z}}\xspace}
\newcommand{\M}{\ensuremath{\mathcal{M}}\xspace}
\newcommand{\A}{\ensuremath{\mathcal{A}}\xspace}
\newcommand{\supp}{\text{supp}}
\newcommand{\ie}{\emph{i.e.}\@\xspace}
\newcommand{\etc}{\emph{etc.}\@\xspace}
\newcommand{\cfgpat}[1]{{}^{\omega}\!#1^{\omega}}
\newcommand{\set}[1]{\left\{#1\right\}}
\newcommand{\seq}{\theta}

\newcommand{\nota}[1]{}
\newcommand{\TMsym}[7]{
	\begin{tikzpicture}[scale=0.9]
		\filldraw[fill=black!10] (0,0) -- (1,0) -- (1,0.5) -- (0,0.5) -- cycle;
		\filldraw[fill=black!10] (0,0.5) -- (1,0.5) -- (1,1) -- (0,1) -- cycle;
		\filldraw[fill=black!10] (0,1) -- (1,1) -- (1,1.5) -- (0,1.5) -- cycle;
		\filldraw[fill=black!10] (0,1.5) -- (1,1.5) -- (1,2) -- (0,2) -- cycle;
		\draw (0.5,0.25) node {#1};
		\draw (0.5,0.75) node {$R$};
		\draw (0.5,1.25) node {$q_0$};
		\draw (0.5,1.75) node {$\mb$};
		\filldraw[fill=black!10] (1,0) -- (2,0) -- (2,0.5) -- (1,0.5) -- cycle;
		\filldraw[fill=black!10] (1,0.5) -- (2,0.5) -- (2,1) -- (1,1) -- cycle;
		\filldraw[fill=black!10] (1,1) -- (2,1) -- (2,1.5) -- (1,1.5) -- cycle;
		\filldraw[fill=black!10] (1,1.5) -- (2,1.5) -- (2,2) -- (1,2) -- cycle;
		\draw (1.5,0.25) node {#2};
		\draw (1.5,0.75) node {#3};
		\draw (1.5,1.25) node {#4};
		\draw (1.5,1.75) node {#5};
		\filldraw[fill=black!10] (-0.5,0) -- (0,0) -- (0,0.5) -- (-0.5,0.5);
		\filldraw[fill=black!10] (-0.5,0.5) -- (0,0.5) -- (0,1) -- (-0.5,1); 
		\filldraw[fill=black!10] (-0.5,1) -- (0,1) -- (0,1.5) -- (-0.5,1.5); 
		\filldraw[fill=black!10] (-0.5,1.5) -- (0,1.5) -- (0,2) -- (-0.5,2);
		\filldraw[fill=black!10] (2.5,0) -- (2,0) -- (2,0.5) -- (2.5,0.5);
		\filldraw[fill=black!10] (2.5,0.5) -- (2,0.5) -- (2,1) -- (2.5,1); 
		\filldraw[fill=black!10] (2.5,1) -- (2,1) -- (2,1.5) -- (2.5,1.5); 
		\filldraw[fill=black!10] (2.5,1.5) -- (2,1.5) -- (2,2) -- (2.5,2);
		\draw (1,-0.5) node {#6};
		\draw[thick,->] (0.45+#7,2.3) -- (0.45+#7,2);
	\end{tikzpicture}
}

\title{
Computational Aspects of Asynchronous CA}
\author{J\'er\^ome Chandesris\inst{2} \and Alberto Dennunzio\inst{2} \and Enrico Formenti\inst{2}\thanks{Corresponding author.} \and Luca Manzoni\inst{1}}
\date{}
\institute{
         Universit\`a degli Studi di Milano--Bicocca\\
        Dipartimento di Informatica, Sistemistica e Comunicazione,\\
        Viale Sarca 336, 20126 Milano (Italy)\\
\email{luca.manzoni@disco.unimib.it}
        \and
Universit\'e Nice-Sophia Antipolis,
Laboratoire I3S,\\
2000 Route des Colles, 06903 Sophia Antipolis (France).\\
\email{{enrico.formenti,alberto.dennunzio}@unice.fr}
          }

\begin{document}

\maketitle

\begin{abstract}
This work studies some aspects of the computational power of fully asynchronous cellular automata (ACA).
We deal with some notions of simulation between ACA and Turing Machines. In particular, we characterize the updating sequences specifying which are ``universal'', i.e., allowing a (specific family of) ACA to simulate any TM on any input. We also consider the computational cost of such simulations.
\end{abstract}

\section{Introduction}
Cellular Automata (CA) are a computational model widely used in many scientific fields.
A CA consists of identical finite automata arranged over a regular lattice. Each automaton updates its state on the basis of its own state and the one of its neighbors. All updates are synchronous. 

CA are particularly successful for modelling real systems~\cite{Chopard11}. However, many natural systems have a clear asynchronous behavior (think of biological processes for example~\cite{Schonfisch99}). Asynchronous CA (ACA) have been introduced in order to be able to more closely simulate these systems. Roughly speaking, they are CA in which the constraint of synchronicity has been relaxed.

According to which updating policy is chosen, the behavior of the ACA under consideration can be very different. In literature, several policies have been considered (purely asynchronous~\cite{Nakamura74}, $\alpha$-asynchronous~\cite{Regnault09}, \etc)

In this paper we consider a fully asynchronous behavior in which, as in a continuous time process, two cells are never updated simultaneously (or, equivalently, one and only one cell is updated at each time step)~\cite{Fates06,Schonfisch99}. Of course, the evolutions of such ACAs depend on the sequence of cells that are updated at each time step.   

It is well-known that CA are capable of universal computation (see~\cite{Ollinger11} for an up-to-date survey). Various mechanisms have been introduced to show how ACA can emulate classical CA or circuits~\cite{Nakamura74,Nehaniv02,Peper05,Worsch10}. As a consequence of these results, Turing universality of ACA is also proved. 

In this paper we focus on the way by which an ACA  simulates a TM.
We analyse two different modes: strict simulation and scattered strict
simulation. Roughly speaking, strict simulations pretend that the ACA
exactly reproduces the steps of the TM (up to some encoding) admitting at most that some time is wasted between two steps of the TM. The second mode is essentially the same as strict simulation but considers the case that only a subset of cells can be used to perform the simulation, the others being ``inactive''. We characterize all updating schemes allowing these two simulation modes. Moreover, we show that in both cases, the time slowdown due to asynchronicity is quadratic w.r.t. the running time of the TM under simulation.

\section{Basic Notions}
Let $A$ be an alphabet. A \emph{CA configuration} (or simply configuration) is a function from $\Z$ to $A$. A \emph{cellular automaton} is a structure $(A, r,\lambda)$ where $A$ is the \emph{alphabet}, $r$ is the  \emph{radius}, and $\lambda:A^{2r+1}\to A$ is the  \emph{local rule}. The local rule is applied synchronously to all positions of a configuration. In other words, the local rule induces a global rule $f:\az\to\az$ defined as follows:
\[
\forall c\in\az, \forall i\in\Z, \quad f(c)_i=\lambda(c_{i-r}, \ldots, c_{i+r})
\]
Where $\az$ is the set of all configurations and for each $i \in \Z$, the image $c(i)$ is denoted by $c_i$. Let $\lambda$ be a local rule of radius $r$. Consider now the following asynchronous updating of a configuration. At each time $t$, $f$ is applied on one and only one position. We denote by $\theta$ the sequence $(\seq_{t})_{t>0}$ whose generic element $\theta_t$ is the position which is updated at time $t$.
\begin{definition}
A \emph{fully asynchronous cellular automaton (ACA)} is a quadruple $(A,\lambda,r,\seq)$ where $A$ is a finite alphabet, $\lambda:A^{2r+1}\mapsto A$
is the local rule of radius $r\in\mathbb{N}$ and $\seq=(\seq_{t})_{t>0}$, with $\seq_{t}\in\mathbb{Z}$
is a sequence of cell positions.
\end{definition}

Every ACA $\mathcal{A} = (A,\lambda,r,\theta)$ induces a dynamical behavior described as follows.
The evolution of any configuration $c\in A^\mathbb{Z}$ is the sequence of configurations 
\[
\lbrace f^{0}(c), f^{1}(c),\, f^{2}(c),\, \ldots, f^{t}(c),  f^{t+1}(c), \ldots\rbrace,
\] where $f^0(c)=c$ and, for any $t\in\N$, $f^{t+1}(c)$ is defined as
\[
\forall i\in\Z,\qquad f^{t+1}(c)_{i}=
\begin{cases}
\lambda(f^{t}(c)_{i-r},\ldots,f^{t}(c)_{i+r}) & \text{if }i=\theta_t\\
c_{i} & \text{otherwise}
\end{cases}
\]
We stress that dynamical evolutions in asynchronous CA depend on the choice of the updating policy. Different updating schemas have been considered for studying asynchronicity in CA settings~\cite{Nakamura74,Fates06,Regnault09}. We deal with the fully asynchronous situation in which at each time only one cell is updated according to $\theta$. Note that even for some relatively simple rules the behavior under fully asynchronous updating is far to be simple, see for instance~\cite{Peper04,Regnault09}. From now on, for a sake of simplicity, by referring to ACA we will mean fully asynchronous CA.    

A  \emph{Turing Machine} \M is a 7-tuple $(Q,\Sigma,\Gamma,\mb,\delta,q_0,F)$, where $Q$ is the set of \emph{states}, $\Sigma \subset \Gamma$ is the \emph{input alphabet}, $\Gamma$ is the \emph{working alphabet} and 
$\mb \in\Gamma\setminus\Sigma$ is the \emph{blank symbol}. The map $\delta: Q \times \Gamma \mapsto Q \times \Sigma \times \{L,R\}$ is the \emph{transition function}, where $L$ and $R$ denote the left and right movements of the head, $q_0 \in Q$ is the \emph{initial state} and $F \subseteq Q$ the \emph{set of final states} (see~\cite{hopcroft79}, for an introduction on this subject). An
\emph{instantaneous configuration} $c$ of \M is a triple $(T,q,p)$, where $T \in \Gamma^\mathbb{Z}$ is the content of the tape, $q \in Q$ is the current state of \M and $p \in \mathbb{Z}$ is the position of 
its head. 

A \emph{run} of \M on the initial input $x\in\Sigma^*$ is
the sequence $R_t=\{(T_t,q_t,p_t)\}_{t\in\Z}$ where 
$(T_0,q_0,p_0)=(^{\omega}\mb x\mb^{\omega},q_0,0)$ (\ie the symbol $\mb$ is repeated infinitely many times at the left and at the right of $x$)for $t=0$, and for any $t\in\N$,
$(T_{t+1},q_{t+1},p_{t+1})$ is the instantaneous configuration of \M at time $t+1$, where $T_{t+1}$ is equal to $T_t$ except that in position $p_t$ in which the symbol $(T_t)_{p_{t}}$ is replaced by the symbol $s$ with $(q_{t+1}, s, X)=\delta(q_t, (T_t)_{p_t})$, where $p_{t+1}=p_t + 1$ if $X=R$ and $p_{t+1} = p_t - 1$ if $X=L$. 
Remark that if in a run it happens that at some time $t\in\N$ the state $q_t\in F$, then, for any $k>t$, $R_k=R_t$, i.e., in other words the computation halts on the instantaneous configuration $R_t$ and the output is the non blank content of tape $R_t$.

\section{Simulation of Turing Machines}
It is well-known that CA are a universal computational model according to different notions of universality (for a survey 
see~\cite{Ollinger11}). The main point to prove
universality is to simulate a TM. Of course, one can apply similar
ideas and constructions to prove computational universality or computational capability of ACA (see for example~\cite{Nakamura74,Nehaniv02,Peper05}). In this section we would like to precise the computational cost of such simulations.
\medskip

The basic idea when simulating a TM using an ACA is to act by ``extracting'' first the information about the state of the TM from the current configuration of the ACA, and then to operate the TM transition saving information on the ACA configuration again. The way of saving the TM state and the way we extract it from the current configuration lead to the two following notions of simulation.

\paragraph{Notation.}
For $a,b\in\N$ with $a<b$, denote $[a,b]$ the set of integers
between $a$ and $b$ (including $a$ and $b$).
Given the finite sets $A_1,A_2,\ldots,A_n$, for any element
$(a_1,a_2,\ldots,a_n)\in A_1\times A_2\times\ldots\times A_n$, 
define $\Pi_{A_i}((a_1,a_2,\ldots,a_i\ldots,a_n))=a_i$.
These projection maps can be naturally extended to work with
configurations, indeed given a configuration 
$c\in(A_1\times\ldots\times A_n)^{\Z}$,
for any $i\in[1,n]$ and $j\in\Z$, define $\Pi_{A_i}(c)_j=\Pi_{A_i}(c_j)$.
\smallskip

Given a configuration $c\in\az$ and a function $\psi:\Z\to\Z$, $c^{\psi}$
is the configuration defined as $c^{\psi}_i=c_{\psi(i)}$ for all $i\in\Z$.
\begin{definition}
Let $\M = (Q,\Sigma,\Gamma,\mb,\delta,q_0,F)$ be a TM and $\A = (A,\lambda,r,\theta)$ be an ACA. \A \emph{strictly simulates} \M iff $A = \Gamma \times B$ for some finite set $B$ and for any input $x\in\Sigma^*$ of \M, there exists a configuration $c\in\az$ satisfying the following conditions:
\begin{enumerate}
  \item $\Pi_{\Gamma}(c)= \cfgpat{
  \mb x\mb}$ and $\Pi_{B}(c)= \cfgpat{sus}$ for some $u,s\in B$;\\
  \item for any time $t\in\N$, there exists $t'\in\N$ such that 
	\[\Pi_{\Gamma}(f^{t'}(c))=T_t,\]
 \item for any pair of times $t_1,t_2\in\N$,
 \[
 t_1<t_2\Rightarrow t^\prime_1<t^\prime_2,
 \] 
 where $t^\prime_i=\min\{k\in\N: \Pi_{\Gamma}(f^{k}(c))=T_{t_i}\}$, $i=1,2$.
\end{enumerate}
\end{definition}
In other words, an ACA \A strictly simulates a TM \M if its configurations can represent in a direct way the tape of \M, possibly
using an additional amount of information (stored in the alphabet B) and some additional time. Relaxing the condition on the representation of the tape, the following weaker notion of simulation is obtained.
\begin{definition}\label{def:weakly}
Let $M = (Q,\Sigma,\Gamma,\mb,\delta,q_0,F)$ be a TM. Let $\A = (A,\lambda,r,\theta)$ be an ACA. \A \emph{scattered strictly simulates} \M iff $A = \Gamma \times B$ for some finite set $B$, and for any input $x\in\Sigma^*$ of \M there exist an injective increasing function $\psi: \mathbb{Z} \mapsto \mathbb{Z}$ and a configuration $c\in\az$ satisfying the following conditions:
\begin{enumerate}
	\item 
	\begin{enumerate}
  \item $\Pi_{B}(c^{\psi})= \cfgpat{sus}$, for some $u,s\in B$\enspace;
  \item $\forall i\in\Z\setminus\psi(\Z), \Pi_{B}(c_i)=q$, for some $q\in B$\enspace;
  \item $\Pi_{\Gamma}(c^\psi) = \cfgpat{\mb x\mb}$\enspace;
  \end{enumerate}
  \item for any time $t\in\N$, there exists $t'\in\N$ such that
	 \[\Pi_{\Gamma}((f^{t'}(c))^{\psi})=T_t\] 
\item 
for any  pair of times $t_1,t_2\in\N$,
 \[
 t_1<t_2\Rightarrow t^\prime_1<t^\prime_2,
 \] 
 where $t^\prime_i=\min\{k\in\N: \Pi_{\Gamma}((f^{k}(c))^\psi )=T_{t_i}\}$, $i=1,2$.
\end{enumerate}
\end{definition}
A scattered strict simulation assumes that only a subset of cells participates to the simulation and the others are somehow inactive. For this reason, in the ACA configurations there can be an offset made for example by $\mb$s between the symbols of the TM tape content.  
Note that when the function $\psi$ is $\psi(i) = i$ then scattered strict simulation and strict simulation coincide.
\smallskip

According to the above definitions, even if an ACA can
simulate a TM on a fixed input $x$, it might not be able to simulate
the same TM on a different input simply because of an inappropriate
updating sequence $\theta$.

\subsection{Construction 1.} Given a TM $\M = (Q,\Sigma,\Gamma,\mb,\delta,q_0,F)$ build a family of ACA $\A_{\theta}=(A,\lambda,1,\theta)$
such that $A=\Gamma \times Q \times D \times C$, where $D=\set{L,R}, C=\set{0,1,2}$, and the local rule
$\lambda: A^3\to A$ is defined as follows
\[
	\lambda(u,v,z) =
		\begin{cases}
			(\sigma,q,m,0) & \begin{array}{l} \text{if } u = (\sigma_u,q_u,R,1), v = (\sigma_v,q_v,m_v,2) \\
			 \text{and } \delta(q_u,\sigma_v) = (q,\sigma,m)\end{array}\\
			 (\sigma_v,q_v,R,2) & \text{if } v= (\sigma_v,q_v,R,1), z = (\sigma_z,q_z,m_z,0) \\
			(\sigma_v,q_v,m_v,1) & \text{if } u = (\sigma_u,q_u,R,2), v = (\sigma_v,q_v,m_v,0) \\

			(\sigma,q,m,0) & \begin{array}{l} \text{if } z = (\sigma_z,q_z,L,1), v = (\sigma_v,q_v,m_v,2) \\
			 \text{and } \delta(q_z,\sigma_v) = (q,\sigma,m)\end{array}\\
			 (\sigma_v,q_v,L,2) & \text{if } v= (\sigma_v,q_v,L,1), u = (\sigma_u,q_u,m_u,0) \\
			(\sigma_v,q_v,m_v,1) & \text{if } z = (\sigma_z,q_z,L,2), v = (\sigma_v,q_v,m_v,0) \\
                           v&\text{otherwise.}				
\end{cases}
\]

\begin{figure}
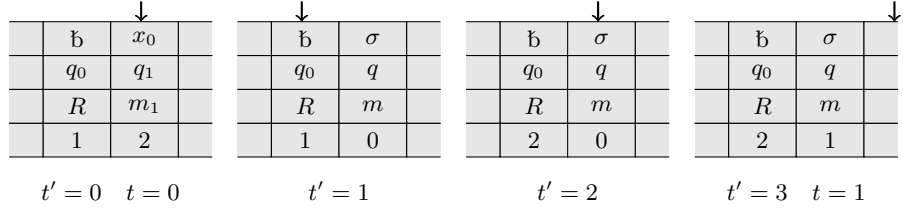

	\begin{center}
	\TMsym{$1$}{$2$}{$m_1$}{$q_1$}{$x_0$}{$t'=0$\quad $t=0$}{1}
	\TMsym{$1$}{$0$}{$m$}{$q$}{$\sigma$}{$t'=1$}{0}
	\TMsym{$2$}{$0$}{$m$}{$q$}{$\sigma$}{$t'=2$}{1}
	\TMsym{$2$}{$1$}{$m$}{$q$}{$\sigma$}{$t'=3$\quad$t=1$}{2}
	\end{center}
	\caption{Simulation of the first step of a TM using an ACA built by construction 1 with updating sequence $\theta= (0,-1,0,1,\ldots)$. The ACA and TM times are denoted by $t'$ and $t$, respectively. The arrow points at the current active cell of the ACA.}
\label{fig:TM_sim}
\end{figure}

Every cell of the ACA contains the symbol of the corresponding cell on the TM tape, the state of the TM, the direction of movement of the TM head, and a value $\xi\in C$ to control the simulation. At TM time $t$, the cell $i$ with $\xi=1$ is the one where the TM head is positioned at time $t-1$, i.e., $p_{t-1}=i$ (with $p_{-1}=-1$). During the ACA evolution, at most one cell in whole configuration has $\xi=1$. If the updating sequence allows the cell $i+1$ (resp., $i-1$) to be updated and the cell $i$ has $m=R$ (resp.,  $m=L$), then the cell $i+1$ (resp., $i-1$) changes its state according to the TM rule and its own value of $\xi$ is set to $0$ to indicate that the information about the head position has to be moved to this cell. To perform it, at subsequent times ACA will set the cell with $\xi=1$ to $2$ and the cell with $\xi=0$ to 1. 
An example of this behavior is shown in Figure \ref{fig:TM_sim}.
\smallskip

In order to (strictly) simulate a TM on input $x=x_0\ldots x_{n-1}\in\Sigma^*$, ACA given by the above construction have to start on the following configuration $c\in A^{\Z}$:
\[
         \forall i\in\Z,\;
	c_i = 
		\begin{cases}
			(x_i,q,m,2) & \text{ if } 0 \le i < |x|. \\
			(\mb,q_0,R,1) & \text{ if } i = -1.\\
			(\mb,q,m,2) & \text{ otherwise,}
		\end{cases}
\]
where $q$ and $m$ are an arbitrarily chosen state and movement (since they will not be used in the simulation, their choice can be arbitrary). The last point to precise is which updating sequences
can be used. Of course, that depends on the TM to simulate but there are sequences that can be used in ``all occasions'', they are called universal. An updating sequence $\theta$ is \emph{universal}  iff
$|\{i\in\Z,\theta_i=k\}|=\infty$ for every $k\in\Z$; practically speaking a sequence is universal if any cell is updated infinitely many times. 

\begin{theorem}
\label{teo:TM_strict_sim}
An ACA $\A = (A,\lambda,1,\theta)$ given by Construction $1$ is Turing universal if and only if $\theta$ is universal.
\end{theorem}
\begin{proof}
Consider a TM $\M = (Q,\Sigma,\Gamma,\delta,\mb, q_0,F)$ and an ACA 
$\A = (A,f,1,\theta)$ with $\theta$ universal. For any $x\in\Sigma^*$ input of \M, let $c$ be the initial configuration built in construction $1$. Let us prove that \A strictly simulates \M.
\smallskip

Let $R_t = (T_t,q_t,p_t)$ be the configuration of the \M at time $t$. We claim that for all $t\in\N$ there exists $t'\in\N$ such that the configuration $c'=f^{t'}(c)$ of \A has the following properties
\begin{enumerate}
	\item $\Pi_{\Gamma}(c') = T_t$.
	\item $\Pi_{Q}(c'_{p_{t-1}}) = q_t$  \quad($p_{-1}=-1$)
	\item $\Pi_{C}(c'_{p_t}) = 1$ and $\Pi_{C}(c'_i)\ne 1$, for all $i\in\Z\setminus\{p_t\}$.
	\item $\Pi_{D}(c'_{p_t}) = R$ if $p_t > p_{t-1}$; $L$ if $p_t < p_{t-1}$. 
\end{enumerate}

We proceed by induction. For $t = 0$ the claim is true by construction ($t'=0$ and $c'=c$). 
Assume that the claim is true for $t>0$, \ie,  there exists $t'$ such that the configuration $c'=f^{t'}(c)$ satisfies the four stated properties. Remark that $\Pi_{C}(c'_{p_t}) = 1$ and hence the only cells that
can change their value are at positions $p_t+1$ or $p_t-1$, depending on the value of $\Pi_{D}(c'_{p_t})$.
Assume that $\Pi_{D}(c'_{p_t})=R$ (the other case is similar). Since $\theta$ is universal, there exists
$t''>t'$ such that $\theta_{t''}=p_t+1$ and for any other $\bar{t}\in\N$ either $\theta_{\bar{t}}\ne p_t+1$ or
$\bar{t}>t''$. According to the definition of $f$, at time $t''$ the cell $p_t+1$ will become
$(\sigma,q,m,0)$, where $(\sigma,q,m)=\delta(\Pi_{Q}(c''_{p_t}), \Pi_{\Gamma}(c''_{p_t+1}))$ and
$c''=f^{t''}(c)$. Moreover, no other cell can change its content between time $t'+1$ and $t''-1$.
Therefore $\Pi_{\Gamma}(c'') = T_{t+1}$ and $\Pi_{Q}(c''_{p_t+1}) = q_{t+1}$ \ie $c''$ satisfies the
first and second properties.

Again, since $\theta$ is universal, there exists
$t'''>t''$ such that $\theta_{t'''}=p_t$ and for any other $\bar{t}\in\N$ either $\theta_{\bar{t}}\ne p_t+1$ or
$\bar{t}>t'''$. According to $f$, the forth component of the cell at position $\theta_{t'''}$ is set to $2$.
Once more, remark that no changes in the configuration of the ACA occur between time
$t''+1$ and $t'''-1$.

Finally, by the universality of $\theta$, there exists $\tilde{t}>t'''$ such that $\theta_{\tilde{t}}=p_t+1$ and for any other $\bar{t}\in\N$ either $\theta_{\bar{t}}\ne p_t+1$ or $\bar{t}>\tilde{t}$. From the definition of $f$, one deduces that the only possibility is that the fourth component of the cell at position $p_t+1$ in $f^{\tilde{t}}(c)$ is set to $1$.  Since no changes in the configuration of the ACA occur between time $t'''+1$ and $\tilde{t}-1$, the claim is proved. 

To prove the inverse implication, consider the TM
\[
\M=(\set{q_R,q_L},\set{0,1},\set{0,1,\mb},\mb, \delta,q_R,\emptyset)
\]
where $\delta$ is defined as follows
\begin{center}
\begin{tabular}{c|c|c|c|c|c|c}
$(q,\sigma)$&$(q_R,\mb)$&$(q_R,0)$& $(q_R,1)$&$(q_L,\mb)$&$(q_L,0)$&$(q_L,1)$ \\
\hline
$\delta(q,\sigma)$&$(q_L,1,L)$&$(q_R,1,R)$&$(q_R,1,R)$&$(q_R,0,R)$&$(q_L,0,L)$&$(q_L,0,L)$\\
\end{tabular}
\end{center}

On any input, \M writes a symbol $1$ on the cell $0$, then it writes $1$s towards the right until a blank symbol is reached. When a blank is reached, it moves left writing a symbol $0$ until a blank is encountered at this point it starts moving right writing $1$s and so forth. It is clear that the head of \M passes through
any cell of the tape infinitely many times. Therefore any ACA given by Construction 1 needs an universal updating sequence to strictly simulate it.\qed
\end{proof}

The previous result proves that the class of ACA given by Construction $1$ are computational universal
but it seems that requiring an universal updating sequence involves a considerable time loss (see the proof of the Theorem\ref{teo:TM_strict_sim}). The following proposition shows that there exist (carefully chosen) updating sequences such that the time loss is acceptable (quadratic).

\begin{proposition}\label{prop:exec1}
Given a TM \M that executes in time $T(n)$, there exists an ACA \A given by Construction 1 that simulates 
\M in time $O\left( T(n)^2 \right)$.
\end{proposition}
\begin{proof}
Let $s_i$ with $i \in\N$ be the finite sequence given by all the integers between $-i$ and $+i$ with step $2$ (e.g., $s_2 = (-2,0,2)$ and $s_1 = (-1,1)$). Consider the sequence $\theta$ given by the concatenation of the $s_i$ sequences:
\[
\theta = \underbrace{s_0 (-1) s_0} \; \underbrace{s_1 s_0 s_1} \; \underbrace{s_2 s_1 s_2} \ldots = \underbrace{0,-1,0},\underbrace{-1,1,0,-1,1},-2,0,2,-1,1,\ldots
\]
Clearly, $\theta$ is universal. Consider $\A = (A,\lambda,1,\theta)$ where $\lambda$ and $A$ are as in Construction $1$.
It is easy to verify that every block $s_i s_{i-1} s_i$ simulates one step of \M. The size of the block increases by a (multiplicative) constant $3$ for every $i$. The total length of the simulation is then bounded by 
\[
1+|s_0|+\sum_{i = 1}^{T(n)} 2\cdot|s_i|+|s_{i-1}|=
2 + \sum_{i = 1}^{T(n)} 3i + 2 = 
\frac{3}{2}\left(T(n)^2 + \frac{7}{3}T(n)+\frac{4}{3}\right) = O\left(T(n)^2\right)
.
\]
\end{proof}
\begin{remark}
The time $O\left(T(n)^2\right)$ is the better asymptotic limit for a sequence. Indeed, consider a sequence such that the corresponding ACA simulates in time $O(g(T(n))) \subseteq O(T(n)^2)$ any TM which executes in time $T(n)$. We show that $O(g(T(n))) = O(T(n)^2)$.
Consider  all the possible movement of the TM head. Its position is $0$ at time $0$, either $1$ or $-1$, at time $1$, and so on. Let us introduce the  graph $G=(V,E)$ where $V=\{0,\ldots,T(n)\}\times\{-T(n),\ldots,T(n)\}$ and $E = \{((t,a),(t+1,b)) \;|\; t \in \{0,\ldots,T(n)-1\} \text{ and } |a-b|=1\}$. Any path starting from $(0,0)$ and ending to $(T(n),b)$ with $b \in \{-T(n),\ldots,T(n)\}$ represents a possible sequence of a TM head. In order to simulate all the head movements, all the nodes of the graph must be visited at least once by these paths. Since the graph has $O(T(n)^2)$ vertexes, the time $O(g(T(n)))$ can not be less than $O(T(n)^2)$.
\end{remark}
\begin{remark}
Actually, for any a TM \M which executes in time $O(T(n))$, there are uncountably many ACA given by Construction 1 that simulate \M in time $O(T(n)^2)$. An infinite set of them is individuated by   the sequences obtained from the one illustrated in the Proposition 1  ``inserting''  in it any integer in one or more positions. 
\end{remark}
The slowdown in the simulation of TM using ACA given by Construction $1$ is essentially given by the
fact that we want a strict simulation and we must keep track (among other things) of the position of the head of the TM. Relaxing this last constraint brings to a different notion of simulation and to a different
construction.
\subsection{Construction 2}
Given a TM $\M = (Q,\Sigma,\Gamma,\mb, \delta,q_0,F)$ build a family of ACA $\A_{\theta}=(A,\lambda,1,\theta)$
such that $A=\Gamma \times Q \times D \times C$,
where $D=\set{L,R}, C=\set{1,2}$ and $\lambda: A^3\to A$ is defined as follows
\[
	\lambda(u,v,z) =
		\begin{cases}
			(\sigma,q,m,1) & 
					\begin{array}{l} 
						\text{if } u = (\sigma_u,q_u,R,1),
						\Pi_{\Gamma}(v)=\sigma_v,\\
						\Pi_C(z)\ne 1
						\text{ and } \delta(q_u,\sigma_v) = (q,\sigma,m)
					\end{array}\\
			(\sigma',q',m',1) & 
				\begin{array}{l}
						\text{if } z = (\sigma_z,q_z,L,1), \Pi_{\Gamma}(v)=\sigma_v,\\
						\Pi_C(u)\ne 1		
						\text{ and } \delta(q_z,\sigma_v) = (q',\sigma',m')
				\end{array}\\
			(\Pi_{\Gamma}(v),\Pi_Q(v),\Pi_D(v),2) & \text{otherwise.} \\
		\end{cases}
\]
The initial configuration for starting the simulation is the same as the one given for Construction $1$. As already remarked before, this
family does not strictly simulate \M since it does not keep track
of the head position. However, using similar techniques
as in Theorem~\ref{teo:TM_strict_sim}, one can prove that for
$\theta=s_0s_1s_2...$, the ACA $(A,\lambda,1,\theta)$ simulates \M
on any input with a total running time
\[
\sum_{i=0}^{T(n)} |s_i|=\sum_{i=0}^{T(n)} (i+1)=\frac{1}{2}
\left(T(n)^2 + 3T(n)+2\right)=
O\left(T(n)^2\right)
\]
where $T(n)$ is the running time of \M on the input $x$.

\subsection{Construction 3}
Construction $1$ and $2$ assume that potentially all cells of the
ACA cooperate to the simulation of the TM. Assume now that only
a subset of cells participates to the simulation and the others are
somehow inactive. In this section, we are going to show that the
ACA can still (scattered strictly) simulate any TM whenever the updating
sequence has some specific properties.

A set $S\subset\Z$ is \emph{syndetic} if there exists some finite $E\subset\Z$ such that
$\cup_{n\in E}(S-n)=\Z$, where $(S-n)=\set{k\in\Z\,|\,k+n\in S}$. Syndetic sets have bounded gaps \ie there exists $g\in\N$ (which depends on $S$) such that for any $h\in\Z$, $\set{h,h+1,\ldots,h+g}\cap S\ne\emptyset$. Given a sequence $\alpha=\set{\alpha_i}_{i\in\Z}$, the \emph{support} of $\alpha$ is the set $\supp(\alpha)=\cup_{i\in\Z}\set{\alpha_i}$.

\emph{Notation.} 
To shorten up the notation in what follows, given an ordered sequence
of states $u^{(-r)},\ldots,u^{(0)},\ldots,u^{(r)}$, denote $E_R(k)=\set{i\in[1,r]\,|\,\Pi_C(u^{(i)})=k}$ and similarly $E_L(k)=\set{i\in[-r,-1]\,|\,\Pi_C(u^{(i)})=k}$. Finally, denote
$j_R=\min E_R(k)$ if $E_R(k)\ne\emptyset$ and  $j_L=\max E_L(k)$ if $E_L(k)\ne\emptyset$.

Given a TM $\M = (Q,\Sigma,\Gamma,\delta,q_0,F)$ build a family of ACA $\A_{\theta}=(A,\lambda,r,\theta)$
such that $A=\Gamma \times Q \times D \times C$,
$D=\set{L,R}$, $C=\set{0,1,2,3}$. The local rule $\lambda: A^{2r+1}\to A$ is defined as follows
\[
	\lambda(u^{(-r)},\ldots,u^{(0)},\ldots,u^{(r)}) =
		\begin{cases}
			(\sigma,q,m,0) & \begin{array}{l} \text{if } E_L(1)\ne\emptyset,\\
                            u^{(j_L)} = (\sigma_u,q_u,R,1),\\
                            u^{(0)} = (\sigma_v,q_v,m_v,2), \\
			 \text{and } \delta(q_u,\sigma_v) = (q,\sigma,m)\end{array}\\[11mm]
			 (\sigma_v,q_v,R,2) & \begin{array}{l}  \text{if } u^{(0)} = (\sigma_v,q_v,R,1), \\
                            E_R(0)\ne\emptyset,\\
                            \text{and } u^{(j_R)} = (\sigma_z,q_z,m_z,0) \end{array}\\[8mm]
			(\sigma_v,q_v,m_v,1) & \begin{array}{l} \text{if }  E_L(2)\ne\emptyset,\\
                            u^{(j_L)} = (\sigma_u,q_u,R,2),\\  
                            u^{(0)} = (\sigma_v,q_v,m_v,0) \end{array}\\[7mm]
			(\sigma,q,m,0) & \begin{array}{l} \text{if }  E_R(1)\ne\emptyset,\\
                           u^{(j_R)} = (\sigma_z,q_z,L,1),\\
                            u^{(0)}  = (\sigma_v,q_v,m_v,2) \\
			 \text{and } \delta(q_z,\sigma_v) = (q,\sigma,m)\end{array}\\[8mm]
			 (\sigma_v,q_v,L,2) &\begin{array}{l}  \text{if } u^{(0)} = (\sigma_v,q_v,L,1), \\
                            E_L(0)\ne\emptyset,\\
                            u^{(j_L)}  = (\sigma_u,q_u,m_u,0)\end{array} \\[8mm]
			(\sigma_v,q_v,m_v,1) & \begin{array}{l} \text{if } E_R(2)\ne\emptyset,\\
                           u^{(j_R)} = (\sigma_z,q_z,L,2),\\
                           u^{(0)} = (\sigma_v,q_v,m_v,0)\end{array}\\[6mm]
                           u^{(0)}&\text{otherwise.}				
\end{cases}
\]
In order to be able to (scattered strictly) simulate a TM on input $x_0\ldots x_{n-1}\in\Sigma^*$, ACA given by the above construction have to be started on the following configuration $c$
\[
        \forall i\in\Z,\,
	c_i^{\alpha} =
	\begin{cases}
		(\mb, q_0, R, 1) & \text{if } i = \alpha_{-1}\\
		(x_i, q, m, 2) & \text{if } i\in[\alpha_0,\alpha_{n-1}]\\
		(\mb, q, m, 2) & \text{if } i\in\supp(\alpha)\setminus[\alpha_0,\alpha_{n-1}]\\
		(\mb,q,m,3)& \text{otherwise}
	\end{cases}
\]
where $\alpha$ is a subsequence of $\theta$ such that $\alpha_0<\alpha_1<\ldots<\alpha_n$,
and $q\in Q$, $m\in D$ are arbitrarily chosen. Clearly, the whole construction (and hence the simulation) depends on $\theta$ and its subsequence
$\alpha$. The following result characterizes them.

\begin{theorem}
An ACA $\A = (A,\lambda,r,\theta)$ given by Construction $3$ scattered strictly
simulates any TM on any input if and only  if $\theta$ contains an universal subsequence
$\alpha$ such that $\supp(\alpha)$ is a syndetic set.
\end{theorem}
\begin{proof}
For any TM $\M = (Q,\Sigma,\Gamma,\mb,\delta, q_0,F)$ consider an ACA 
$\A = (A,\lambda,r,\theta)$ given by Construction $3$.  Assume that \A scattered strictly simulates \M.
First of all, let us prove that $|\supp(\theta)|=\infty$. Indeed, if $|\supp(\theta)|<\infty$, only
a finite number of cells can be used for simulation and therefore, according to Condition $2$ of
Definition\ \ref{def:weakly} only a finite portion of the tape can be simulated.
Choose a subsequence $\alpha$ of $\theta$ such that $\alpha_0<\alpha_1<\ldots<\alpha_n$ (this is possible since $|\supp(\theta)|=\infty$). By contradiction, assume that no $\alpha$ is  universal. This means that the set of cells that can be updated infinitely many times is finite or empty. Without loss of generality
assume that it has finite cardinality and $j>0$ is the maximal of its elements (the case $j\leq 0$ is similar).
Let $k$ be the index of last occurrence of $j$ in $\alpha$. Consider the TM \M from the proof of
Theorem~\ref{teo:TM_strict_sim} on the empty input. Since, for all $t>k$, $f^t(c^{\alpha}))_j=c^{\alpha}_j$, Condition $2$ of Definition~\ref{def:weakly} is violated.

Now, always by contradiction, assume that there exist universal subsequences of $\theta$ but none of them has a syndetic support set. This means that there are
larger and larger sets $[a,b]\subset\N$ not contained in $\supp(\alpha)$. Choose one of them such that $b-a>r$ and let \M be the TM which on the empty input writes $2b$ symbols $1$ as an output. Set
$h=\min_{j\in\supp(\alpha)}\set{b<j}$. According to the definition of $\lambda$, for all
$t\in\N$, $f^t(c^{\alpha}))_h=c^{\alpha}_h$. Hence Condition $2$ of Definition~\ref{def:weakly} is false.
Therefore if \A scattered strictly simulates \M, $\theta$ has to contain an universal subsequence whose support is
syndetic.

On the other hand, assume that $\theta$ contains an universal subsequence $\alpha$ with
syndetic support. Then, there exists a subsequence $\alpha'$ with
$\alpha'_0<\alpha'_1<\ldots<\alpha'_n$, where $n$ is the length of the input of \M. Build the initial
configuration $c^{\alpha'}$ as described in Construction $3$. Since $\supp(\alpha)$ has bounded gaps, let $p$ be the shortest one and
set the radius $r=p$. The rest of the proof is essentially the
same as the one given for Theorem~\ref{teo:TM_strict_sim}
with $\alpha'$ playing the role of $\theta$.\qed
\end{proof}

\begin{proposition}\label{prop:exec3}
For any TM \M that executes in time $T(n)$, consider an ACA \A given by Construction 3 and an updating sequence $\theta$ containing
an universal subsequence whose support is syndetic. Then, \A
scattered strictly simulates  \M in time $O\left( T(n)^2 \right)$.
\end{proposition}
\begin{proof}
Consider $\A = (A,\lambda,p,\theta)$ where $\lambda$ and $A$ are as in Construction $3$. 
For $i \in\N$, let $s_i$ be as in the proof of 
Proposition~\ref{prop:exec1}. Let $\alpha$ be the subsequence of
$\theta$ given in the hypothesis and let $p\in\N$ be the minimal gap.
Consider the subsequence
$
\alpha' = s_{2p}s_{4p}\ldots s_{2ip}\ldots
$
Similarly to the proof of Proposition~\ref{prop:exec1},
$s_{2ip}s_{4ip}s_{6ip}$ can be used to encode the simulation
of a step of \M. The total length of the simulation is then bounded by 
\[
\sum_{i = 1}^{T(n)} |s_{2ip}|+|s_{4ip}|+|s_{6ip}|=
\sum_{i = 1}^{T(n)} 12ip + 3 = 
6p\left(T(n)^2 + (1+\frac{1}{2p})T(n)\right) = O\left(T(n)^2\right).
\]\qed
\end{proof}

\section{Future work}
This paper studies the simulation of Turing Machines by fully asynchronous CA. We have shown that computational universality for ACA corresponds to the condition that the ACA updating sequence contains any cell infinitely many times (universal sequence). We have also exhibited some universal sequences in order that the computational cost of the simulation is reasonable (quadratic) w.r.t the length of the run of the simulated TM.

The construction of sequences that allow quadratic simulations might
seem artificial, classifying our results in the domain of computability but of no use in practical simulations. Indeed, for example, consider the program described in~\cite{amar04}. It is used to simulate biochemical processes in cells. It is essentially based on an ACA (although the authors do not clearly state this) which represents proteins (and other chemicals) by particles. The current configuration is updated in two steps: diffusion and collision/reaction. In the diffusion
step a particle is randomly chosen, a direction is randomly chosen and then the particle makes a move in this direction. Adopting a more ``system'' based vision instead of a particle based one and passing in 1D for simplicity sake, one can represent the sequence of activations
of the different sites (which may contain a particle or not) as
$\theta_0=0,\theta_1=\theta_0+X
_1,\theta_2=\theta_1+X_2,\ldots,\theta_t=\theta_{t-1}+X_t,\ldots$, where $X_i$ are random variables with values in $\set{-1,+1}$ identically distributed (with uniform Bernoulli distribution for example).
Therefore, the current active cell $\theta_t=\sum_{i=0}X_i$ is also a random variable with values in \Z and, practically speaking, the updating sequence is a random walk (see~\cite{Kin93} for more on random walks). It is well-known that 1D random walks pass through all sites infinitely many times, therefore $\theta$ is an universal updating sequence and hence, by Theorem~\ref{teo:TM_strict_sim}, the system descrived in~\cite{amar04} is capable of Turing universal computation.

For every alphabet $A$, local rule $\lambda$ and radius $r$ it is possible to consider the class of all the ACA such in the form $(A,\lambda,r,\seq)$ where $\seq$ is generated by a 1D random walk. Then it is possible to investigate for any property the probability of finding an ACA inside that class that has the given property. In the case of simulation of Turing machines it can be proved that
for any TM \M running in time $T(n)$, there exists an ACA with updating sequences generated by random walks such that the probability of (strictly) simulating \M in time $3T(n)$ is $2^{-3T(n)}$.
The above result says that there exist a class of ACA where a particular ACA \A chosen uniformly from that class has a low probability of simulating a TM faster than deterministic ACA seen in the paper. The authors are currently investigating to see if and up to which extent the above results can be improved.

\section*{Acknoledgements}
The authors warmly thanks Thomas Worsch for pointing out relevant literature. 
\nocite{*}
\bibliographystyle{plain}
\bibliography{article.bib}

\end{document}